\documentclass[fleqn]{llncs}

\usepackage{url}
\usepackage{hhline}
\usepackage{ifthen}
\usepackage{latexsym}
\usepackage{amsmath}
\usepackage{amssymb}
\usepackage{stmaryrd}
\usepackage{xspace}
\usepackage{wrapfig}
\usepackage{subfig}
\usepackage{floatflt}
\usepackage{pstricks,pst-node,pst-tree}
\usepackage{graphicx}
\usepackage{algorithm}
\usepackage{algorithmic}

\renewcommand{\emptyset}{\varnothing}	
\renewcommand{\phi}{\varphi}    	

\newcommand{\ebox}{\quad{\vrule height4pt width4pt depth0pt}}

\newcommand{\itename}{\fn{ite}\xspace}
\newcommand{\mkndname}{\fn{mknd}\xspace}
\newcommand{\bddorname}{\fn{or}\xspace}

\newcommand{\randbddname}{\fn{gen\_rand\_bdd}\xspace}
\newcommand{\makerandbddname}{\fn{rand\_bdd}\xspace}
\newcommand{\randsinkname}{\fn{rand\_sink}\xspace}

\newcommand{\cointossname}{\fn{cointoss}\xspace}
\newcommand{\toaffname}{\fn{to\_aff}\xspace}
\newcommand{\getmodelname}{\fn{get\_model}\xspace}
\newcommand{\xorclosename}{\fn{xor\_close}\xspace}
\newcommand{\translatename}{\fn{translate}\xspace}
\newcommand{\travname}{\fn{trav}\xspace}
\newcommand{\extendname}{\fn{extend}\xspace}
\newcommand{\consname}{\fn{cons}\xspace}

\newcommand{\ite}[3]{\ensuremath{\itename(#1,#2,#3)}}
\newcommand{\mknd}[3]{\ensuremath{\mkndname(#1,#2,#3)}}
\newcommand{\randbdd}[2]{\ensuremath{\randbddname(#1,#2)}}
\newcommand{\toaff}[1]{\ensuremath{\toaffname(#1)}}
\newcommand{\getmodel}[1]{\ensuremath{\getmodelname(#1)}}
\newcommand{\xorclose}[1]{\ensuremath{\xorclosename(#1)}}
\newcommand{\translate}[2]{\ensuremath{\translatename(#1,#2)}}
\newcommand{\trav}[3]{\ensuremath{\travname(#1,#2,#3)}}
\newcommand{\extend}[3]{\ensuremath{\extendname(#1,#2,#3)}}
\newcommand{\cons}[4]{\ensuremath{\consname(#1,#2,#3,#4)}}
\newcommand{\makerandbdd}[3]{\ensuremath{\makerandbddname(#1,#2,#3)}}
\newcommand{\randsink}[0]{\ensuremath{\randsinkname}}

\newcommand{\cointoss}[0]{\ensuremath{\cointossname()}}
\newcommand{\bddtrue}{\ensuremath{1}\xspace}
\newcommand{\bddfalse}{\ensuremath{0}\xspace}
\newcommand{\true}[0]{\id{1}\xspace}
\newcommand{\false}[0]{\id{0}\xspace}

\newcommand{\bxor}{\ensuremath{\oplus}}

\newcommand{\B}{{\ensuremath{\cal B}}\xspace}

\newcommand{\Var}{{\ensuremath{\cal V}}\xspace}

\newcommand{\Valuations}{{\ensuremath{\cal I}}\xspace}
\newcommand{\PartialValuations}{{\ensuremath{{\cal I}_\mathsf{p}}}\xspace}

\newcommand{\vars}[1]{\ensuremath{\mathsf{vars}(#1)}}
\newcommand{\bits}[1]{\ensuremath{\mathsf{bits}(#1)}}
\newcommand{\aff}[1]{\ensuremath{\mathsf{aff}(#1)}}

\newcommand{\Bool}{\classname{B}\xspace}

\newcommand{\zeroval}{\vec{0}}

\newcommand{\bddor}[2]{\ensuremath{\bddorname(#1,#2)}}

\renewcommand{\lnot}[1]{\ensuremath{\overline{#1}}}
\newcommand{\doubleb}[1]{[\![ #1 ]\!]}

\newcommand{\id}[1]{\textit{#1}}
\newcommand{\fn}[1]{\textsf{#1}}
\newcommand{\classname}[1]{\textbf{#1}}

\newcommand{\defined}[1]{\textit{#1}}
\newcommand{\etal}{\textit{et al.}}

\newpsobject{showgrid}{psgrid}{subgriddiv=1,griddots=10,gridlabels=6pt}

\title{Binary Decision Diagrams \\ for Affine Approximation}
\author{Kevin Henshall, Peter Schachte, Harald S{\o}ndergaard and Leigh Whiting}
\institute{Department of Computer Science and Software Engineering \\ The University of Melbourne, Vic.\ 3010, Australia}
\date{}

\begin{document}

\maketitle
\begin{abstract}
Selman and Kautz's work on ``knowledge compilation'' 
established how approximation (strengthening and/or
weakening) of a propositional knowledge-base can be used to
speed up query processing, at the expense of completeness.
In this classical approach, querying uses Horn over-
and under-approximations of a given knowledge-base,
which is represented as a propositional formula in
conjunctive normal form (CNF).  Along with the class of
Horn functions, one could imagine other Boolean function
classes that might serve the same purpose, owing to
attractive deduction-computational properties similar
to those of the Horn functions.  Indeed, Zanuttini has
suggested that the class of affine Boolean functions could
be useful in knowledge compilation and has presented an
affine approximation algorithm.  Since CNF is awkward
for presenting affine functions, Zanuttini considers both a
sets-of-models representation and the use of modulo 2
congruence equations.
In this paper, we propose an algorithm
based on reduced ordered binary decision diagrams (ROBDDs).
This leads to a representation which is more compact than
the sets of models and, once we
have established some useful properties of affine Boolean
functions, a more efficient algorithm.
\end{abstract}

\section{Introduction}
\label{sec:intro}

A recurrent theme in artificial intelligence is the efficient
use of (propositional) knowledge-bases.
A promising approach, which was
initially proposed by Selman and Kautz~\cite{Sel-Kau:JACM96},
is to query (and perform deductions from)
upper and lower approximations, commonly called \emph{envelopes} and
\emph{cores} respectively, of a given knowledge-base.
By choosing approximations that allow more efficient inference, it
is often possible to quickly determine that the envelope of the 
given knowledge-base entails the query, and therefore so does the full
knowledge-base, avoiding the costly inference from
the full knowledge-base.  When this fails, it may be possible to quickly show
that the query is not entailed by the core, and therefore 
not entailed by the full knowledge-base.  Only when both of these fail must
the full knowledge-base be used for inference.

It is usually assumed that Boolean functions are
represented in clausal form, and that approximations are
Horn~\cite{Sel-Kau:JACM96,delVal:AI2005}, as inference from Horn
knowledge-bases is exponentially more efficient than from unrestricted
knowledge-bases.
However, it has been noted that there are other well-understood
classes that have computational properties that include some of
the attractive properties of the Horn class. 

Zanuttini~\cite{Zan:ECAI02,Zan:SARA02} discusses
the use of other classes of Boolean functions for approximation
and points out that \emph{affine} approximations have
certain advantages over Horn approximations, most notably the
fact that they do not blow out in size.
This is certainly the case when affine functions are represented 
in the form of modulo-2 congruence equations.
The more general sets-of-models representation is also considered
by Zanuttini.
In this paper, we consider another general representation, namely the
well-known \emph{Reduced Ordered Binary Decision Diagrams (ROBDDs)}.
We prove some important properties of affine
functions represented as ROBDDs,
and present a new ROBDD algorithm for deriving affine envelopes.  

The balance of the paper proceeds as follows.
In Section~\ref{sec:background} we 
recapitulate the definition of the Boolean affine class, and we establish
some of their important properties.
We also briefly introduce ROBDDs, but mainly to fix our notation,
as we assume that the reader is familiar with Boolean
functions and their representation as decision diagrams.
Section~\ref{sec:algorithms} recalls the model-based affine envelope algorithm,
and develops our own ROBDD-based algorithm, along with a correctness
proof.
Section~\ref{sec:eval} describes our testing methodology, including our
algorithm for generating random ROBDDs, and presents our results.
Section~\ref{sec:conclusion} discusses related work and applications,
and concludes.

\section{Propositional Classes, Approximation and ROBDDs}
\label{sec:background}

We use ROBDDs~\cite{bra-rud-dac-90,bry-toc-86} to represent
Boolean functions.
Horiyama and Ibaraki~\cite{Hor-Iba:AI2003} have recommended ROBDDs
as suitable for implementing knowledge bases.
Our choice of ROBDDs as a data structure was not so
much influenced by that recommendation, as by the convenience
of working with a canonical representation for Boolean functions,
and one that lends itself to inductive reasoning and recursive
problem solving.
Additionally, ROBDD-based inference is fast, and in particular, checking
whether a valuation is a model of an $n$-place function given by an ROBDD
requires a path traversal of length no more than $n$.

ROBDD algorithms for approximation are of interest in their own 
right and some find applications in dataflow 
analysis~\cite{Sch-Son:VMCAI06}.
From this aspect, this paper continues earlier work by
Schachte and S{\o}ndergaard~\cite{Sch-Son:VMCAI06,Sch-Son:SARA07}
who gave algorithms for finding monotone, Krom, and Horn
envelopes.
Here we introduce an ROBDD algorithm for affine envelopes, which
is new.

\subsection{Boolean functions}
\label{sec:prelim}

Let $\B = \{\false,\true\}$ and let $\Var$ be a denumerable
set of variables.
A \emph{valuation} $\mu:\Var\rightarrow \B$ is a (total) assignment of truth
values to the variables in $\Var$.
Let $\Valuations = \Var\rightarrow \B$ denote the set of
$\Var$-valuations.
A \emph{partial valuation} $\mu:\Var\rightarrow \B \cup \{\bot\}$ assigns
truth values to some variables in $\Var$, and $\bot$ to others.  Let
$\PartialValuations = \Var\rightarrow \B \cup \{\bot\}$.
We use the notation $\mu[x \mapsto i]$, where $x\in \Var$ and $i\in \B$,
to denote the valuation $\mu$ updated to map $x$ to $i$, that is,
\[
  \mu[x \mapsto i](v) = \left\{
        \begin{array}{ll}
        i & \mbox{\quad if $v = x$} \\
        \mu(v) & \mbox{\quad otherwise}
        \end{array}
  \right.
\]
A Boolean function over ${\Var}$ is a
function $\varphi: \Valuations \to \B$.
We let $\Bool$ denote the set of all Boolean functions over $\Var$.
The ordering on $\B$ is the usual: $x \leq y$ iff $x = \false \lor y = \true$.
\Bool is ordered pointwise, so that the ordering relation corresponds exactly
to classical entailment, $\models$.
It is convenient to overload the symbols for truth and falsehood.
Thus we let $\true$ denote the largest element of \Bool (that is,
$\lambda \mu.\true$) as well as of $\B$.
Similarly $\false$ denotes the smallest element of \Bool
(that is, $\lambda \mu.\false$) as well as of $\B$.
A valuation $\mu$ is a
\emph{model} for $\varphi$, denoted $\mu \models \varphi$, if
$\varphi(\mu)=\true$.
We let $\id{models}(\varphi)$ denote the set of models of $\varphi$.
Conversely, the unique Boolean function that has exactly the set
$M$ as models is denoted $\id{fn}(M)$.
A Boolean function $\phi$ is said to be \emph{independent of} a variable $x$
when for all valuations $\mu$, $\mu[x\mapsto 0] \models \phi$ iff
$\mu[x\mapsto 1] \models \phi$.

In the context of an ordered set of $k$ variables of interest,
$x_1, \ldots, x_k$, we may identify with $\mu$ the binary
sequence $\bits{\mu}$ of length $k$:
\[
  \mu(x_1), \ldots, \mu(x_k)
\]
which we will write simply as a bit-string of length $k$.
Similarly we may think of, and write, the set of valuations
$M$ as a set of bit-strings:
\[
  \bits{M} = \{\bits{\mu} \mid \mu \in M \}
\]
As it hardly creates confusion, we shall present
valuations variously as functions or bitstrings.
We denote the \emph{zero valuation}, which maps
$x_i$ to $\false$ for all $1 \leq i \leq k$, by $\zeroval$.

We use the Boolean connectives $\neg$ (negation),
$\land$ (conjunction), $\lor$ (disjunction) and $+$ 
(exclusive or, or ``xor'').
These connectives operate on Boolean functions, that is,
on elements of $\Bool$.
Traditionally they are overloaded to also operate on
truth values, that is, elements of $\B$.
However, we deviate at this point, as the distinction
between xor and its ``bit-wise'' analogue will be
critical in what follows.
Hence we denote the $\B$ (bit) version by $\bxor$.
We extend this to valuations and bit-strings in the natural way:
\[
  (\mu_1 \bxor \mu_2) (x) = \mu_1(x) \bxor \mu_2(x)
\]
and we let ${}\bxor_3$ denote the ``xor of three'' operation
$\lambda \mu_1 \mu_2 \mu_3 . \mu_1 \bxor \mu_2 \bxor \mu_3$.
We follow Zanuttini~\cite{Zan:ECAI02} in further overloading 
`$\bxor$' and using the notation
\[
  M_\mu = \mu \bxor M = \{\mu \bxor \mu' \mid \mu' \in M\}
\]
We read $M_\mu$ as ``$M$ translated by $\mu$''.
Note that for any set $M$, the function $\lambda \mu . M_\mu$
is an involution: $(M_\mu)_\mu = M$.

A final overloading results in the following definition.
For $\varphi \in \Bool$, and $\mu \in \Valuations$, let
$
  \varphi \bxor \mu = \id{fn}(M_\mu)
$
where $M = \id{models}(\varphi)$.

\subsection{The affine class}
\label{sec:classes}

An \emph{affine} function is one whose set of models is closed under
pointwise application of $\bxor_3$~\cite{Schaefer1978:complexity}.
Affine functions have a number of attractive properties, as we
shall see.
Syntactically, a Boolean function is affine
iff it can be written as a conjunction of affine equations
\[
  c_1 x_1 + c_2 x_2 + \ldots + c_k x_k = c_0
\]
where $c_i \in \{0,1\}$ for all $i \in \{0,..,k\}$.\footnote{In
some circles, such as
cryptography/coding community, the term ``affine'' is used only 
for a function that is 0 or 1, or can be written 
$
  c_1 x_1 + c_2 x_2 + \ldots + c_k x_k + c_0
$
(the latter is what Post~\cite{Post:func-compl41} 
called an ``alternating'' function).
The resulting set of ``affine'' functions is not closed under conjunction.
}
This is well known, but for completeness we prove it below.

The affine class contains $\true$ and is closed under conjunction.
Hence the concept of a (unique) affine envelope is well defined,
and the operation of taking the affine envelope is an
upper closure operator~\cite{Sch-Son:VMCAI06}.
For convenience, let us introduce a name for this operator:
\begin{definition} \rm
Let $\varphi$ be a Boolean function.
The \emph{affine envelope}, $\id{aff}(\varphi)$, of $\varphi$
is defined:
\[
  \id{aff}(\varphi) = 
    \bigwedge \{\psi \mid \mbox{$\varphi \models \psi$ and $\psi$ is affine}\}
  \ebox
\]
\end{definition}
There are numerous other classes of interest, including 
isotone, antitone, Krom, Horn,
$k$-Horn~\cite{Dec-Pea:AI92}, and $k$-quasi-Horn functions,
for which the concept of an envelope is well-defined, as they
form upper closure operators~\cite{Sch-Son:SARA07}.\footnote{Popular classes
such as \emph{unate} functions and \emph{renamable Horn} are not
closed under conjunction and therefore do not have well-defined
concepts of (unique) envelopes.
For example, $x \rightarrow y$ and
$x \leftarrow y$ both are unate, while $x \leftrightarrow y$ is not,
so the ``unate envelope'' of the latter is not well-defined.}

Zanuttini~\cite{Zan:ECAI02} exploits the close connection between 
vector spaces and the sets of models of affine functions.
For our purposes we call a set $S$ of bistrings a \emph{vector space}
iff $\zeroval \in S$ and $S$ is closed under $\bxor$.
The next proposition
simplifies the task of doing model-closure under $\bxor_3$.

\begin{proposition}[\cite{Zan:ECAI02}] \rm
\label{prop:translate}
A non-empty set of models $M$ is closed under $\bxor_3$ iff
$M_\mu$ is a vector space,
where $\mu$ is any element of $M$.
\end{proposition}
\begin{proof}
Let $\mu$ be an arbitrary element of $M$.
Clearly $M_\mu$ contains $\zeroval$, so the right-hand side of
the claim amounts to $M_\mu$ being closed under $\bxor$.

For the `if' direction, assume $M_\mu$ is closed under $\bxor$
and consider $\mu_1, \mu_2, \mu_3 \in M$.
Since $\mu \bxor \mu_2$ and $\mu \bxor \mu_3$ are in $M_\mu$, so is
$\mu_2 \bxor \mu_3$.
And since furthermore $\mu \bxor \mu_1$ is in $M_\mu$, so is
$\mu \bxor \mu_1 \bxor \mu_2 \bxor \mu_3$.
Hence $\mu_1 \bxor \mu_2 \bxor \mu_3$ is in $M$.

For the `only if' direction, assume $M$ is closed under $\bxor_3$,
and consider $\mu_1, \mu_2 \in M_\mu$.
All of $\mu, \mu \bxor \mu_1$ and $\mu \bxor \mu_2$ are in $M$,
and so $\mu \bxor (\mu \bxor \mu_1) \bxor (\mu \bxor \mu_2) = 
\mu \bxor \mu_1 \bxor \mu_2 \in M$.
Hence $\mu_1 \bxor \mu_2 \in M_\mu$.
\ebox
\end{proof}
\begin{proposition}
\rm
A Boolean function is affine iff it can be written as
a conjunction of equations
\[
  c_1 x_1 + c_2 x_2 + \ldots + c_k x_k = c_0
\]
where $c_i \in \{0,1\}$ for all $i \in \{0,..,k\}$.
\end{proposition}
\begin{proof}
Assume the Boolean function $\varphi$ is given as a conjunction of 
equations of the indicated form and let $\mu_1$, $\mu_2$ and $\mu_3$ 
be models.
That is, for each equation we have
\[
\begin{array}{l}
  c_1 \mu_1(x_1) + c_2 \mu_1(x_2) + \ldots + c_k \mu_1(x_k) = c_0 \\
  c_1 \mu_2(x_1) + c_2 \mu_2(x_2) + \ldots + c_k \mu_2(x_k) = c_0 \\
  c_1 \mu_3(x_1) + c_2 \mu_3(x_2) + \ldots + c_k \mu_3(x_k) = c_0
\end{array}
\]
Adding left-hand sides and adding right-hand sides, making use
of the fact that `$\cdot$' distributes over `$+$', we get
\[
  c_1 \mu(x_1) + c_2 \mu(x_2) + \ldots + c_k \mu(x_k) = c_0 + c_0 + c_0 = c_0
\]
where $\mu = \mu_1 \bxor \mu_2 \bxor \mu_3$.
As $\mu$ thus satisfies each equation, $\mu$ is a model of $\varphi$.
This establishes the `if' direction.

For the `only if' part, note that
by Proposition~\ref{prop:translate}, we obtain a vector space
$M_\mu$ from any non-empty set $M$ closed under $\bxor_3$ 
by translating each element of $M$ by $\mu \in M$.
A basis for $M_\mu$ can be formed by taking one vector at a time from $M_\mu$ 
and adding it to the basis if it is linearly independent of the existing basis 
vectors.  
From this basis, a set of linear equations 
\[
  \begin{array}{ccccccc}
    a_{11} x_1 & \bxor & \cdots & \bxor & a_{1k} x_k & = & 0
 \\ a_{21} x_1 & \bxor & \cdots & \bxor & a_{2k} x_k & = & 0
 \\ \vdots     &       &        &       &            & \vdots &
 \\ a_{j1} x_1 & \bxor & \cdots & \bxor & a_{jk} x_k & = & 0
  \end{array}
\]
can be computed that have exactly $M_\mu$ as their set of models
(a method is provided by Zanuttini~\cite{Zan:ECAI02}, 
in the proof of his Proposition 3).
Each function 
$f_i = \lambda x_1,\ldots,x_k.a_{i1} x_1 \bxor \cdots \bxor a_{ik} x_k$
is linear, so for $\nu \in M_\mu$, 
$f_i(\nu \bxor \mu) = f_i(\nu) + f_i(\mu) = f_i(\mu)$.
Hence $M$ can be described by the set of affine equations
\[
  \begin{array}{ccccccc}
    a_{11} x_1 & \bxor & \cdots & \bxor & a_{1k} x_k & = & f_1(\mu)
 \\ a_{21} x_1 & \bxor & \cdots & \bxor & a_{2k} x_k & = & f_2(\mu)
 \\ \vdots     &       &        &       &            & \vdots &
 \\ a_{j1} x_1 & \bxor & \cdots & \bxor & a_{jk} x_k & = & f_j(\mu)
  \end{array}
\]
as desired.
\ebox
\end{proof}
It follows from the syntactic characterisation that
the number of models possessed by an affine function is 
either 0 or a power of 2.
Other properties will now be established that are used in the 
justification of the affine envelope algorithm of Section~\ref{sec:algorithms}.
The first property is that if a Boolean function $\varphi$ has two
models that differ for exactly one variable $v$, then its affine
envelope will be independent of $v$.
To state this precisely we introduce a concept of a
``characteristic'' valuation for a variable.
\begin{definition} \rm
In the context of a set of variables $V$, let $v \in V$.
The \emph{characteristic valuation} for $v$, $\chi_v$, is defined by
\[
  \chi_v(x) = \left\{ \begin{array}{ll}
                        \true & \mbox{ if $x = v$} \\
                        \false & \mbox{ otherwise}  \ebox
                      \end{array}
              \right.
\]
\end{definition}
Note that $\mu \bxor \chi_v$ is the valuation which agrees with
$\mu$ for all variables except $v$.
Moreover, if $\mu \models \varphi$, then both of $\mu$ and
$\mu \bxor \chi_v$ are models of $\exists v (\varphi)$.

\begin{proposition} \rm
\label{prop:vs}
Let $\varphi$ be a Boolean function whose set of models forms
a vector space, and assume that for some valuation $\mu$ and some variable $v$,
$\mu$ and $\mu \bxor \chi_v$ both satisfy $\varphi$.
Then $\varphi$ is independent of $v$.
\end{proposition}
\begin{proof}
The set $M$ of models contains at least two elements, and since it is closed 
under $\bxor$, $\chi_v$ is a model.
Hence for \emph{every} model $\nu$ of $\varphi$, $\nu \bxor \chi_v$ is
another model.
It follows that $\varphi$ is independent of $v$.
\ebox
\end{proof}

\begin{proposition} \rm
\label{prop:kill-var}
Let $\varphi$ be a Boolean function.
If, for some valuation $\mu$ and some variable $v$,
$\mu$ and $\mu \bxor \chi_v$ both satisfy $\varphi$,
then $\id{aff}(\varphi) = \exists v (\id{aff}(\varphi))$.
\end{proposition}
\begin{proof}
Let $\mu$ be a model of $\varphi$, with $\mu \bxor \chi_v$
also a model.
For \emph{every} model $\nu$ of $\varphi$, we have that
$\nu \bxor \mu \bxor (\mu \bxor \chi_v)$ satisfies $\id{aff}(\varphi)$,
that is, $\nu \bxor \chi_v \models \id{aff}(\varphi)$.
Now since both $\nu$ and $\nu \bxor \chi_v$ satisfy $\id{aff}(\varphi)$,
it follows that $\exists v (\id{aff}(\varphi))$ cannot have a model
that is not already a model of $\id{aff}(\varphi)$
(and the converse holds trivially).
Hence $\id{aff}(\varphi) = \exists v (\id{aff}(\varphi))$.
\ebox
\end{proof}

\begin{proposition} \rm
\label{prop:aff-exists-commute}
For all Boolean functions $\varphi$,
$\id{aff}(\exists v (\varphi)) = \exists v (\id{aff}(\varphi))$.
\end{proposition}
\begin{proof}
We need to show that the models of $\id{aff}(\exists v (\varphi))$
are exactly the models of $\exists v (\id{aff}(\varphi))$.
Clearly $\id{aff}(\exists v (\varphi))$ is $\false$ iff
$\varphi$ is $\false$ iff $\exists v (\id{aff}(\varphi))$ is $\false$.
So we can assume that $\id{aff}(\exists v (\varphi))$ is
satisfiable---let $\mu \models \id{aff}(\exists v (\varphi))$.
Then, for some positive odd number  $k$, 
\[
  \mu = \mu_1 \bxor \mu_2 \bxor \cdots \bxor \mu_k
\]
with $\mu_1, \ldots, \mu_k$ being different models of $\exists v (\varphi)$.
These $k$ models can be partitioned into two sets, according as they
satisfy $\varphi$; let
\[
  M = \{ \mu_i \mid 1 \leq i \leq k, \mu_i \models \varphi \} \qquad
  M' = \{ \mu_i \mid 1 \leq i \leq k, \mu_i \not\models \varphi \}
\]
Then both $M$ and $M'_{\chi_v}$ consist entirely of models of $\varphi$.
Hence, depending on the parity of $M$'s cardinality, either $\mu$ or
$\mu \bxor \chi_v$ is a model of $\id{aff}(\varphi)$ (or both are).
In either case, $\mu \models \exists v (\id{aff}(\varphi))$.

Conversely, let $\mu \models \exists v (\id{aff}(\varphi))$.
Then either $\mu$ or $\mu \bxor \chi_v$ is a model of $\id{aff}(\varphi)$
(or both are).
Hence $\mu$ (or $\mu \bxor \chi_v$ as the case may be) can be
written as a sum of $k$ models $\mu_1, \ldots, \mu_k$ ($k$ odd)
of $\varphi$.
It follows that both $\mu_1 \bxor \mu_2 \bxor \cdots \bxor \mu_k$
and $\mu_1 \bxor \mu_2 \bxor \cdots \bxor \mu_k \bxor \chi_v$
are models of $\exists v (\varphi)$.
Hence $\mu \models \id{aff}(\exists v (\varphi))$.
\ebox
\end{proof}

\subsection{ROBDDs}
\label{sec:robdds}

We briefly recall the essentials of ROBDDs~\cite{Bryant:CS92}.
Let the set $\Var$ of propositional variables be equipped with a
total ordering $\prec$.
\defined{Binary decision diagrams} (\defined{BDDs}) are defined
inductively as follows:
\begin{itemize}
\item
\bddfalse is a BDD.
\item
\bddtrue is a BDD.
\item
If $x \in \Var$ and $R_1$ and $R_2$ are BDDs then
$\ite{x}{R_1}{R_2}$ is a BDD.
\end{itemize}
Let $R = \ite{x}{R_1}{R_2}$.
We say a BDD $R'$ \defined{appears in} $R$ iff $R' = R$ or $R'$
appears in $R_1$ or $R_2$.  We define
$\vars{R} = \{v \mid \ite{v}{\_}{\_} \mbox{ appears in } R\}$.
The meaning of a BDD is given as follows.
\[ \begin{array}{ll}
    \doubleb{\bddfalse}
        &= \false
 \\ \doubleb{\bddtrue}
        &= \true
 \\ \doubleb{\ite{x}{R_1}{R_2}}
        &= (x \land \doubleb{R_1}) \lor (\lnot{x} \land \doubleb{R_2})
\end{array} \]
A BDD is an \defined{Ordered binary decision diagram} (\defined{OBDD}) iff it 
is \bddfalse or \bddtrue or if it is
\ite{x}{R_1}{R_2}, $R_1$ and $R_2$ are OBDDs, and $\forall x' \in
\vars{R_1} \cup \vars{R_2}: x \prec x'$.

An OBDD $R$ is a \defined{Reduced Ordered Binary Decision Diagram} 
(\defined{ROBDD} \cite{bry-toc-86,Bryant:CS92})
iff for all BDDs $R_1$ and $R_2$ appearing in $R$, $R_1 = R_2$ when
$\doubleb{R_1} = \doubleb{R_2}$.
Practical implementations \cite{bra-rud-dac-90} use a function
\mknd{x}{R_1}{R_2} to create all ROBDD nodes as follows:
\begin{enumerate}
\item If $R_1 = R_2$, return $R_1$ instead of a new node, as
  $\doubleb{\ite{x}{R_1}{R_2}} = \doubleb{R_1}$.
\item If an identical ROBDD was previously built, return that one instead of
  a new one; this is accomplished by keeping a hash table, called
  the \emph{unique table}, of all previously created nodes.
\item Otherwise, return \ite{x}{R_1}{R_2}.
\end{enumerate}
This ensures that ROBDDs are strongly canonical: a shallow equality test is
sufficient to determine whether two ROBDDs represent the same Boolean
function.

Figure~\ref{fig:robdd2} shows some example ROBDDs.
The ROBDD in Figure~\ref{fig:robdd2}(a) denotes the function which has
five models:
$\{00011, 00110, 01001, 01101, 10101\}$.
In general we depict the ROBDD $\ite{x}{R_1}{R_2}$ as a directed
acyclic graph rooted in $x$,
with a solid arc from $x$ to the dag for $R_1$ and a dashed line
from $x$ to the dag for $R_2$.
However, to avoid unnecessary clutter, we omit the node (sink)
for $\false$ and all arcs leading to that sink.

As a typical example of an ROBDD algorithm, Algorithm~\ref{alg:bdd_or} 
generates the disjunction of two given ROBDDs.
This operation will be used by the affine approximation
algorithm presented in Section~\ref{sec:algorithms}.

\begin{algorithm}[t]
\hspace*{2em} \bddor{\bddtrue}{\_} = \bddtrue \\
\hspace*{2em} \bddor{\bddfalse}{\_} = \bddfalse \\
\hspace*{2em} \bddor{\_}{\bddtrue} = \bddtrue \\
\hspace*{2em} \bddor{\_}{\bddfalse} = \bddfalse \\
\hspace*{2em} \bddor{\ite{x}{T}{E}}{\ite{x'}{T'}{E'}} \\
\hspace*{4em} $ \mid x \prec x' = 
\mknd{x}{\bddor{T}{\ite{x'}{T'}{E'}}}{\bddor{E}{\ite{x'}{T'}{E'}}}$ \\
\hspace*{4em} $ \mid x' \prec x = 
\mknd{x'}{\bddor{\ite{x}{T}{E}}{T'}}{\bddor{\ite{x}{T}{E}}{E'}}$ \\
\hspace*{4em} $ \mid \mathbf{otherwise} = 
\mknd{x}{\bddor{T}{T'}}{\bddor{E}{E'}}$ \\
\caption{The ``or'' operator for ROBDDs\label{alg:bdd_or}}
\end{algorithm}

Algorithm~\ref{alg:get_model} is used to extract a model from an ROBDD.  
For an unsatisfiable ROBDD (that is, $\bddfalse$) we return $\bot$.
Although presented here in recursive fashion, it is better implemented in an 
iterative manner whereby we traverse through the ROBDD, one pointer moving down 
the ``else'' branch at each node, a second pointer trailing immediately behind.
If a \bddtrue sink is found, we return the 
path traversed thus far and note that any further variables which we are yet to 
encounter may be assigned any value. If a \bddfalse sink is found, we use the
trailing pointer to
step up a level, follow the ``then'' branch for one step and continue searching 
for a model by following ``else'' branches. This method relies on the fact 
that ROBDDs are ``reduced'', so that if no $\bddtrue$ sink can be 
reached from a node, then the node itself is the $\bddfalse$ sink.

\begin{algorithm}[t]
\hspace*{2em} $\getmodel{\bddfalse} = \bot$ \\
\hspace*{2em} $\getmodel{\bddtrue} = \lambda v. \bot $ \\
\hspace*{2em} $\getmodel{\ite{x}{T}{E}} = $ \\
\hspace*{4em} $ \mathbf{let}$ $\mu$ = \getmodel{T} \\
\hspace*{4em} $ \mathbf{in} $ \\
\hspace*{6em} $ \mathbf{if} $ $\mu$ = $\bot$ $\mathbf{then} $ \\
\hspace*{8em} $ \getmodel{E}[x \mapsto 0] $ \\
\hspace*{6em} $ \mathbf{else}$ $\mu[x \mapsto 1] $ \\
\caption{get\_model algorithm for ROBDDs\label{alg:get_model}}
\end{algorithm}

We shall later use the following obvious corollary of Proposition~\ref{prop:vs}:
\begin{corollary} \rm
\label{cor:vs_robbd}
Let ROBDD $R$ represent a function whose set of models form a vector
space.
Then every path from $R$'s root node to the $\true$-sink contains the
same sequence of variables, namely $\vars{R}$ listed in variable order.
\ebox
\end{corollary}
It is important to take advantage of \emph{fan-in} to create efficient ROBDD
algorithms.  
Often some ROBDD nodes will appear multiple times in a given ROBDD, 
and algorithms that traverse that ROBDD will meet these nodes multiple times.
Many algorithms can avoid repeated work by keeping a cache
of previously seen inputs and their corresponding outputs, called a
\emph{computed table}, see Brace \etal~\cite{bra-rud-dac-90} for details.

\section{Finding Affine Envelopes for ROBDDs}
\label{sec:algorithms}

Zanuttini~\cite{Zan:ECAI02} gives an algorithm, here presented as 
Algorithm~\ref{alg:Zan}, for finding the affine envelope, assuming a 
Boolean function $\varphi$ is represented as a set of models.
This algorithm is justified by Proposition~\ref{prop:translate}.

\begin{algorithm}[t]
\begin{algorithmic}
\STATE {\bf Input:} The set $M$ of models for function $\varphi$. 
\STATE {\bf Output:} $\aff{M}$ --- the set of models of $\varphi$'s affine envelope.
\\
\IF{$M = \emptyset$} 
\STATE {\bf return} $M$ 
\ENDIF
\STATE $N \leftarrow \emptyset$ 
\STATE choose $\mu \in M$ 
\STATE $\id{New} \leftarrow M_\mu$ 
\REPEAT
\STATE $N \leftarrow N \cup \id{New}$ 
\STATE $\id{New} \leftarrow \{ \mu_1 \bxor \mu_2 \mid \mu_1, \mu_2 \in N\} \setminus N$
\UNTIL{$\id{New} = \emptyset$}
\STATE {\bf return} $N_\mu$
\end{algorithmic}
\caption{The sets-of-models based affine envelope algorithm\label{alg:Zan}}
\end{algorithm}
\begin{example}
To see Algorithm~\ref{alg:Zan} in action,
assume that $\varphi$ has four models, 
$M = \{01011, 01100, 10111, 11001\}$, and refer to Figure~\ref{fig:ZanAlg}.
\begin{figure}[t]
\[
  \begin{array}{c}
    M = \left\{ \begin{array}{c}
        01011 \\
        01100 \\
        10111 \\
        11001
        \end{array} \right\} \\~\\~\\~\\
    \mu = 01100
  \end{array}
\quad
  \begin{array}{c}
  M_\mu = \left\{ \begin{array}{c}
        00111 \\
        00000 \\
        11011 \\
        10101
        \end{array} \right\}\\~\\~\\~\\~\\
  \end{array}
\quad
  N     = \left\{ \begin{array}{c}
        00111 \\
        00000 \\
        11011 \\
        10101 \\
        11100 \\
        10010 \\
        01110 \\
        01001
        \end{array} \right\}
\quad
    N_\mu = \aff{M} = \left\{ \begin{array}{c}
        01011 \\
        01100 \\
        10111 \\
        11001 \\
        10000 \\
        11110 \\
        00010 \\
        00101
        \end{array} \right\}
\]
\caption{Steps in Algorithm~\ref{alg:Zan}\label{fig:ZanAlg}}
\end{figure}
We randomly pick $\mu = 01100$ and obtain $M_\mu$ as shown.
The first round of completion under `$\bxor$' adds three bit-strings:
$\{11100, 10010, 01110\}$,
and another round adds $01001$ to produce $N$.
Finally, ``adding back'' $\mu = 01100$ yields the affine
envelope $N_\mu = \aff{M}$.
\ebox
\end{example}

We are interested in developing an algorithm for ROBDDs.
We can improve on Algorithm~\ref{alg:Zan} and at the same 
time make it more suitable for ROBDD manipulation.
The idea is to build the result $N$ step by step, by picking the models $\nu$ 
of $M_\mu$ one at a time and computing $N := N \cup N_\nu$
at each step.
We can start from $N = \{ \zeroval \}$, as $\zeroval$ has to be in $M_\mu$.
This leads to Algorithm~\ref{alg:Zan2}.

\begin{algorithm}[t]
\begin{algorithmic}
\STATE {\bf Input:} The set $M$ of models for function $\varphi$. 
\STATE {\bf Output:} $\aff{M}$ --- the set of models of $\varphi$'s affine envelope. 
\\
\IF{$M =\emptyset$} 
\STATE {\bf return} $M$ 
\ENDIF
\STATE $N \leftarrow \{ \zeroval \}$
\STATE {\bf choose} $\mu \in M$
\STATE $R \leftarrow M_\mu \setminus \{ \zeroval \}$ 
\FORALL{$\nu \in R$}
\STATE $N \leftarrow N \cup N_{\nu}$
\ENDFOR
\STATE {\bf return} $N_\mu$
\end{algorithmic}
\caption{A variant of Algorithm~\ref{alg:Zan}\label{alg:Zan2}}
\end{algorithm}
This formulation is well suited to ROBDDs, as the operation $N_{\nu}$,
that is, taking the xor of a model $\nu$ with each model of the \emph{ROBDD} 
$N$ can be implemented by traversing $N$ and, for each $v$-node with $\nu(v) = 
\true$, swapping that node's children.
And we can do better, utilising two observations.

First, during its construction, there is no need to traverse the ROBDD $N$
for each individual model $\nu$.
A full traversal of $N$ will find all its models systematically,
eliminating a need to remove them one by one.

Second, the ROBDD being constructed can be simplified aggressively
during its construction, by utilising Propositions \ref{prop:kill-var} and
\ref{prop:aff-exists-commute}.
Namely, as we traverse ROBDD $R$ systematically, many paths from the root to the
$\true$-sink will be found that do not contain every variable in $\vars{R}$.
Each such path corresponds to a model \emph{set} of cardinality $2^k$,
$k$ being the number of ``skipped'' variables.
Proposition~\ref{prop:kill-var} tells us that, eventually, the affine envelope
will be independent of all such ``skipped'' variables, and
Proposition~\ref{prop:aff-exists-commute} guarantees that variable
elimination can be interspersed arbitrarily with the process of 
``xoring'' models, that is, we can eliminate variables aggressively.

This leads to Algorithm~\ref{alg:to_aff}.
\begin{algorithm}[t]
\rm
{\bf Input:} An ROBDD $R$. \\
{\bf Output:} The affine envelope of $R$. \\
\\
\hspace*{2em} \toaff{\bddfalse} = \bddfalse \\
\hspace*{2em} \toaff{R} = {\bf let} $\mu = \getmodel{R}$ {\bf in}
		\translate{\xorclose{\translate{R}{\mu}}}{\mu} \\
\\
\hspace*{2em} \translate{\bddfalse}{\_} = \bddfalse \\
\hspace*{2em} \translate{\bddtrue}{\_} = \bddtrue \\
\hspace*{2em} \translate{\ite{x}{T}{E}}{\mu} \\
\hspace*{4em} $\mid (\mu(x) = 0) = 
\cons{x}{\translate{T}{\mu}}{\translate{E}{\mu}}{\mu}$ \\
\hspace*{4em} $\mid (\mu(x) = 1) = 
\cons{x}{\translate{E}{\mu}}{\translate{T}{\mu}}{\mu}$ \\
\\
\hspace*{2em} \xorclose{R} = \trav{R}{\lambda v.\bot}{\bigwedge \{\bar{v} \mid 
v \in \vars{R}\}} \\
\\
\hspace*{2em} $\trav{\bddfalse}{\_}{S} = S$ \\
\hspace*{2em} $\trav{\bddtrue}{\mu}{S}$ \\
\hspace*{4em} $ \mid (\mu \models S) = S$ \\
\hspace*{4em} $ \mid \mathbf{otherwise} = \extend{S}{S}{\mu}$ \\
\hspace*{2em} $\trav{\ite{x}{T}{E}}{\mu}{S} = \trav{T}{\mu[x \mapsto 
1]}{\trav{E}{\mu[x \mapsto 0]}{S}}$ \\
\\
\hspace*{2em} \cons{x}{T}{E}{\mu} \\
\hspace*{4em} $\mid (\mu(x) = \bot) = \bddor{T}{E}$ \\
\hspace*{4em} $\mid \mathbf{otherwise} = \mknd{x}{T}{E}$ \\
\\
\hspace*{2em} \extend{\bddtrue}{\_}{\_} = \bddtrue \\
\hspace*{2em} \extend{\_}{\bddtrue}{\_} = \bddtrue \\
\hspace*{2em} \extend{\bddfalse}{S}{\mu} = \translate{S}{\mu} \\
\hspace*{2em} \extend{\ite{x}{T}{E}}{\bddfalse}{\mu} = 
\cons{x}{\extend{T}{\bddfalse}{\mu}}{\extend{E}{\bddfalse}{\mu}}{\mu} \\
\hspace*{2em} \extend{\ite{x}{T}{E}}{\ite{x}{T'}{E'}}{\mu} \\
\hspace*{4em} $\mid (\mu(x) = 1)$ = 
\mknd{x}{\extend{T}{E'}{\mu}}{\extend{E}{T'}{\mu}} \\
\hspace*{4em} $\mid \mathbf{otherwise} $ = 
\cons{x}{\extend{T}{T'}{\mu}}{\extend{E}{E'}{\mu}}{\mu} \\
\caption{Affine envelopes for ROBDDs\label{alg:to_aff}}
\end{algorithm}
The algorithm combines several operations in an 
effort to amortise their cost. 
We present it in Haskell style, using pattern matching and guarded
equations.
In what follows we step through the details of the algorithm.

The \toaffname function finds an initial model
$\mu$ of $R$, before translating $R$, through the call to \translatename.  
This initial call to \translatename has the effect of ``xor-ing'' $\mu$ with 
all of the models of $R$. Once translated, the xor closure is taken, before 
translating again using the initial model $\mu$ to obtain the affine closure.

\translatename is responsible for computing the xor of a model with an 
ROBDD. Its operation relies on the observation that for a given node $v$ in the 
ROBDD, if $\mu(v) = 1$, then the operation is equivalent to exchanging the 
``then'' and ``else'' branches of $v$.

\xorclosename is used to compute the xor-closure of an ROBDD $R$. 
The third argument passed to \travname is an accumulator in which the 
result is constructed.
As in Algorithm~\ref{alg:Zan2}, we know that $\zeroval$ will be a
model of the result, so we initialise the accumulator as
(the ROBDD for) $\bigwedge \{\bar{v} \mid v \in \vars{R}\}$.

\travname implements a recursive traversal of the ROBDD, and when a model is
 found in $\mu$, we ``extend'' the affine envelope to include the newly found model.
Namely, $\extend{R}{S}{\mu}$ produces (the ROBDD for) $R \lor S_\mu$. Note that 
once a model is found during the traversal, \travname checks if $\mu$ is 
already present within the xor-closure, and if it is not, invokes \extendname 
accordingly. This simple check avoids making unnecessary calls to \extendname.

The \consname function represents a special case of \mkndname. It takes 
an additional argument in $\mu$ and uses it to determine whether to restrict 
away the corresponding node being constructed. The correctness of
\consname rests on 
Propositions \ref{prop:kill-var} and \ref{prop:aff-exists-commute},
which guarantee that affine approximation can be interspersed with
variable elimination, so that the latter can be performed aggressively.

Finally, once a model is found during a traversal, \extendname 
is used to build up the affine closure of the ROBDD. 
In the context of the initial call $\extend{S}{S}{\mu}$,
Corollary~\ref{cor:vs_robbd} ensures that the pattern of the last equation for
\extendname is sufficient: If neither argument is a sink, the two
will have the same root variable.

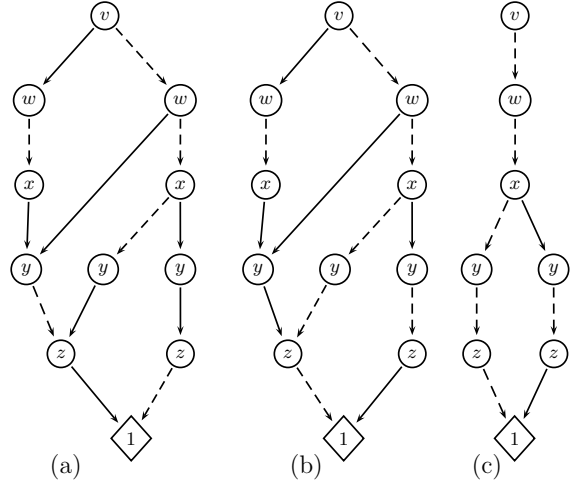
\begin{figure}[t]
\begin{center}
\scalebox{0.8}{
\begin{pspicture}(-2,0)(11,-8)     
\psset{mcol=l,colsep=6pt,nodesep=9pt}
\rput(2,-7.5){\large{(a)}}
\rput(6,-7.5){\large{(b)}}
\rput(9,-7.5){\large{(c)}}

\hfil
\pstree[levelsep=40pt,linestyle=none]{\Tcircle[name=AV,linestyle=solid] {$v$}}{
		\pstree{\Tcircle[name=AW1,linestyle=solid] {$w$}} {
				\pstree{\Tcircle[name=AX1,linestyle=solid] {$x$}} {
						\TR {~~~}
						\pstree{\Tcircle[name=AY1,linestyle=solid] {$y$}} {
							\TR {~~~}
							\pstree{\Tcircle[name=AZ1,linestyle=solid] {$z$}} {
									\TR {~~~}
									\TR {~~~}
									\Tdia [name=ATRUE, linestyle=solid] 
									{$\bddtrue$}
									}
								}
						\pstree{\Tcircle[name=AY2,linestyle=solid] {$y$}} {}
						}
                }
		\pstree{\Tcircle[name=AW2, linestyle=solid] {$w$}}{
				\pstree{\Tcircle[name=AX2,linestyle=solid] {$x$}}{
					\pstree{\Tcircle[name=AY3,linestyle=solid] {$y$}} {
						\pstree{\Tcircle[name=AZ2,linestyle=solid] {$z$}} {}
						}
					}
                }
        }
\psset{nodesepA=1pt,nodesepB=2pt,linecolor=black,arrows=->}
\ncline[linestyle=solid]{AV}{AW1}
\ncline[linestyle=dashed]{AV}{AW2}
\ncline[linestyle=dashed]{AW1}{AX1}
\ncline[linestyle=solid]{AW2}{AY1}
\ncline[linestyle=dashed]{AW2}{AX2}
\ncline[linestyle=solid]{AX1}{AY1}
\ncline[linestyle=dashed]{AX2}{AY2}
\ncline[linestyle=solid]{AX2}{AY3}
\ncline[linestyle=dashed]{AY1}{AZ1}
\ncline[linestyle=solid]{AY2}{AZ1}
\ncline[linestyle=solid]{AY3}{AZ2}
\ncline[linestyle=solid]{AZ1}{ATRUE}
\ncline[linestyle=dashed]{AZ2}{ATRUE}

\hfil
\psset{mcol=l,colsep=6pt,nodesep=9pt}
\pstree[levelsep=40pt,linestyle=none,arrows=-]{\Tcircle[name=BV,linestyle=solid] {$v$}}{
	\pstree{\Tcircle[name=BW1,linestyle=solid] {$w$}} {
		\pstree{\Tcircle[name=BX1,linestyle=solid] {$x$}} {
			\TR {}
			\pstree{\Tcircle[name=BY1,linestyle=solid] {$y$}} {
				\TR {}
				\pstree{\Tcircle[name=BZ1,linestyle=solid] {$z$}} {
					\TR {}
					\TR {}
					\Tdia [name=BTRUE,linestyle=solid] {$\bddtrue$}
					}
				}
			\pstree{\Tcircle[name=BY2,linestyle=solid] {$y$}} {}
			}
		}
	\pstree{\Tcircle[name=BW2, linestyle=solid] {$w$}}{
		\pstree{\Tcircle[name=BX2,linestyle=solid] {$x$}}{
			\pstree{\Tcircle[name=BY3,linestyle=solid] {$y$}} {
				\pstree{\Tcircle[name=BZ2,linestyle=solid] {$z$}} {}
				}
			}
		}
	}
\psset{nodesepA=1pt,nodesepB=2pt,linecolor=black,arrows=->}
\ncline[linestyle=solid]{BV}{BW1}
\ncline[linestyle=dashed]{BV}{BW2}
\ncline[linestyle=dashed]{BW1}{BX1}
\ncline[linestyle=solid]{BW2}{BY1}
\ncline[linestyle=dashed]{BW2}{BX2}
\ncline[linestyle=solid]{BX1}{BY1}
\ncline[linestyle=dashed]{BX2}{BY2}
\ncline[linestyle=solid]{BX2}{BY3}
\ncline[linestyle=solid]{BY1}{BZ1}
\ncline[linestyle=dashed]{BY2}{BZ1}
\ncline[linestyle=dashed]{BY3}{BZ2}
\ncline[linestyle=dashed]{BZ1}{BTRUE}
\ncline[linestyle=solid]{BZ2}{BTRUE}

\hfil
\pstree[levelsep=40pt,linestyle=none,arrows=-]{\Tcircle[name=CV,linestyle=solid] {$v$}}{
        \pstree{\Tcircle[name=CW,linestyle=solid] {$w$}} {
                \pstree{\Tcircle[name=CX,linestyle=solid] {$x$}} {
                        \pstree{\Tcircle[name=CY1,linestyle=solid] {$y$}} {
                                \pstree{\Tcircle[name=CZ1,linestyle=solid] {$z$}} {
                                        \TR {~~~}
										\Tdia [name=CTRUE,linestyle=solid] 
										{$\bddtrue$}
                                        }
                                }
                        \pstree{\Tcircle[name=CY2,linestyle=solid] {$y$}} {
                                \pstree{\Tcircle[name=CZ2,linestyle=solid] {$z$}} {}
                                }
                        }
                }
        }
\psset{nodesepA=1pt,nodesepB=2pt,linecolor=black,arrows=->}
\ncline[linestyle=dashed]{CV}{CW}
\ncline[linestyle=dashed]{CW}{CX}
\ncline[linestyle=dashed]{CX}{CY1}
\ncline[linestyle=solid]{CX}{CY2}
\ncline[linestyle=dashed]{CY1}{CZ1}
\ncline[linestyle=dashed]{CY2}{CZ2}
\ncline[linestyle=dashed]{CZ1}{CTRUE}
\ncline[linestyle=solid]{CZ2}{CTRUE}
\hfil
\end{pspicture}
}
\end{center}
\caption{(a): An example ROBDD $R$; 
note that all our ROBDD diagrams leave out the $\false$-sink and all arcs to it.
(b): The translated version $R_\mu$.  
(c): The vector space $S$ that has been extended to cover $00101$.\label{fig:robdd2}}
\end{figure}

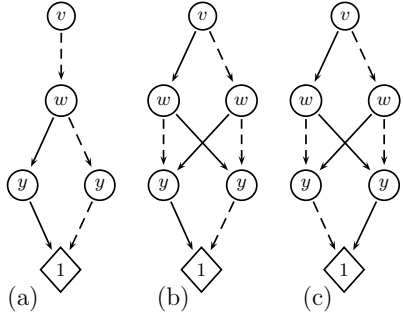
\begin{figure}[t]
\scalebox{0.8}{
\begin{pspicture}(-4,1)(9,-5)  
\rput(0.8,-4.7){\large{(a)}}
\rput(3.3,-4.7){\large{(b)}}
\rput(5.7,-4.7){\large{(c)}}

\psset{mcol=l,colsep=6pt,nodesep=9pt}
\pstree[levelsep=40pt,linestyle=none,arrows=-]{\Tcircle[name=DV,linestyle=solid] {$v$}}{
        \pstree{\Tcircle[name=DW,linestyle=solid] {$w$}} {
                \pstree{\Tcircle[name=DY1,linestyle=solid] {$y$}} {
                        \TR {~~~}
						\Tdia [name=DTRUE,linestyle=solid] {$\bddtrue$}
                        }
                \pstree{\Tcircle[name=DY2,linestyle=solid] {$y$}} {}
                }
        }
\psset{nodesepA=1pt,nodesepB=2pt,linecolor=black,arrows=->}
\ncline[linestyle=dashed]{DV}{DW}
\ncline[linestyle=solid]{DW}{DY1}
\ncline[linestyle=dashed]{DW}{DY2}
\ncline[linestyle=solid]{DY1}{DTRUE}
\ncline[linestyle=dashed]{DY2}{DTRUE}

\pstree[levelsep=40pt,linestyle=none,arrows=-]{\Tcircle[name=EV,linestyle=solid] {$v$}}{
		\pstree{\Tcircle[name=EW1,linestyle=solid] {$w$}} {
				\pstree{\Tcircle[name=EY1,linestyle=solid] {$y$}} {
                        \TR {~~~}
						\Tdia [name=ETRUE,linestyle=solid] {$\bddtrue$}
                        }
				}
		\pstree{\Tcircle[name=EW2,linestyle=solid] {$w$}} {
				\pstree{\Tcircle[name=EY2,linestyle=solid] {$y$}} {}
				}
        }
\ncline[linestyle=solid]{EV}{EW1}
\ncline[linestyle=dashed]{EV}{EW2}
\ncline[linestyle=dashed]{EW1}{EY1}
\ncline[linestyle=solid]{EW1}{EY2}
\ncline[linestyle=solid]{EW2}{EY1}
\ncline[linestyle=dashed]{EW2}{EY2}
\ncline[linestyle=solid]{EY1}{ETRUE}
\ncline[linestyle=dashed]{EY2}{ETRUE}

\pstree[levelsep=40pt,linestyle=none,arrows=-]{\Tcircle[name=FV,linestyle=solid] {$v$}}{
		\pstree{\Tcircle[name=FW1,linestyle=solid] {$w$}} {
				\pstree{\Tcircle[name=FY1,linestyle=solid] {$y$}} {
                        \TR {~~~}
						\Tdia [name=FTRUE,linestyle=solid] {$\bddtrue$}
                        }
				}
		\pstree{\Tcircle[name=FW2,linestyle=solid] {$w$}} {
				\pstree{\Tcircle[name=FY2,linestyle=solid] {$y$}} {}
				}
        }
\ncline[linestyle=solid]{FV}{FW1}
\ncline[linestyle=dashed]{FV}{FW2}
\ncline[linestyle=dashed]{FW1}{FY1}
\ncline[linestyle=solid]{FW1}{FY2}
\ncline[linestyle=solid]{FW2}{FY1}
\ncline[linestyle=dashed]{FW2}{FY2}
\ncline[linestyle=dashed]{FY1}{FTRUE}
\ncline[linestyle=solid]{FY2}{FTRUE}

\end{pspicture}
}
\caption{(a): The vector space $S$ after being extended to cover 
$01$\footnotesize{X}\normalsize$10$. (b): $S$ after extending to cover $10110$.  
(c): $S$ translated to give the affine closure of $R$.\label{fig:robdd3}}

\end{figure}

\begin{example}
Consider the ROBDD $R$ (shown again in Figure~\ref{fig:robdd2}(a)), 
whose set of models is $\{00011,00110,01001,01101,10101\}$.
Picking $\mu = 00011$ and translating gives $R_\mu$, shown in
Figure~\ref{fig:robdd2}(b).
This ROBDD represents a set of vectors $\{00000,00101,01010,01110,10110\}$
which is to be extended to a vector space.

The algorithm now builds up $S$, the xor-closure of $R_\mu$, by taking one 
vector $v$ at a time from $R_\mu$ and extending $S$ to a vector space that 
includes $v$.  $S$ begins as the zero vector.

The first step of the algorithm just adds $00101$ to the existing zero vector
(Figure~\ref{fig:robdd2}(c)).
The next step comes across the vector $01X10$ (which actually represents two 
valuations) and existentially quantifies away the variable $x$
(Figure~\ref{fig:robdd3}(a)).  
Note that the variable $z$ also disappears: this is due to the extension 
required to include $01X10$ that adds enough valuations such that $z$ is 
``covered'' by the vector space.

Extending to cover $10110$ 
simply requires every model to be copied, with $v$ mapped to \true
(Figure~\ref{fig:robdd3}(b)).
Finally, translating back by $\mu$ produces $A$, the affine closure of $R$,
shown in Figure~\ref{fig:robdd3}(c).
\ebox
\end{example}

\section{Experimental Evaluation}
\label{sec:eval}

To evaluate Algorithms~\ref{alg:Zan} and \ref{alg:to_aff} we
generated random Boolean functions using Algorithm~\ref{alg:rand}. 
We generated random Boolean functions of $n$ variables, with an additional
parameter to control the density of the generated function,
that is, to set the likelihood of a random valuation being a model.
For Algorithm~\ref{alg:Zan} we extracted models from the generated
ROBDDs, so that both algorithms were tested on identical Boolean
functions.

\begin{algorithm}[t]
\rm
{\bf Input:} The number $n$ of variables in the random function, \\
\hphantom{{\bf Output:} } $pr$ a calibrator set so that the probability \\
\hphantom{{\bf Output:} } of a valuation being a model is $2^{-pr}$. \\
{\bf Output:} A random Boolean function represented as an ROBDD. \\
\\
\hspace*{2em} \randbdd{n}{pr} = \makerandbdd{0}{n - 1}{pr} \\
\\
\hspace*{2em} $\makerandbdd{m}{n}{pr}$ \\
\hspace*{4em} $\mid (m = n) = \mknd{m}{\randsink}{\randsink}$ \\
\hspace*{4em} $\mid \mathbf{otherwise} = \mknd{m}{T}{E}$ \\
\hspace*{6em} $\mathbf{where}$ \\
\hspace*{8em} $T$ = {\bf if} $(m > n - pr) \land \cointoss$ {\bf then} $\makerandbdd{m + 1}{n}{pr}$ {\bf else} $\bddfalse$ \\
\hspace*{8em} $E$ = {\bf if} $(m > n - pr) \land \cointoss$ {\bf then} $\makerandbdd{m + 1}{n}{pr}$ {\bf else} $\bddfalse$ \\
\\
\hspace*{2em} \randsink = {\bf if} $\cointoss$ {\bf then} $\bddfalse$ {\bf else} $\bddtrue$ \\
\\
\hspace*{2em} $\cointoss$ returns $\true$ or $\false$ with equal probability.
\caption{Generation of random Boolean functions as ROBDDs\label{alg:rand}}
\end{algorithm}

\randbdd{n}{pr} builds, as an ROBDD $R$, a random Boolean function with the 
property that the likelihood of an arbitrary valuation satisfying $R$ is
$2^{-pr}$.
It invokes \makerandbdd{0}{n - 1}{pr}. 
This recursive algorithm builds a 
ROBDD of $(n - pr)$ variables and at depth $(n - pr)$, a random choice 
is made as to whether to continue generating the random function or to simply 
join the branch with a \bddfalse sink. 
If the choice is to continue, then the algorithm recursively 
applies \makerandbdd{m + 1}{n}{pr} to the branch.

By building a ``complete'' ROBDD of $(n - pr)$ variables, we were able to distribute 
the number of models for a given number of variables. In this way, we were able 
to compare the various algorithms for differing model distributions.

Table~\ref{tbl:aff_results} shows the average time (in 
milliseconds) taken by each of the algorithms over 10,000 repetitions 
with the probability 1/1024 of a valuation being a model.
Timing data were collected on a machine running Solaris 9, with two 
Intel Xeon CPUs running at 2.8GHz and 4GB of memory.
Only one CPU was used and tests were run under minimal load on the system.
Our implementation of Algorithm~\ref{alg:Zan} uses sorted arrays of bitstrings
(so that search for models is logarithmic).
As the number of models grows exponentially with the number of variables,
it is not surprising that memory consumption exceeded available space,
so we were unable to collect timing data for more than 15 variables.

\begin{table}[t]
\begin{center}
\begin{tabular}{|c|c|r|}
\hline
    ~~Variables~~ & ~~Algorithm~\ref{alg:Zan}~~ & ~~Algorithm~\ref{alg:to_aff}~~ 
\\
    \hline
    12 & 0.021 & 0.017~~~~~ \\
    15 & 5.991 & 0.272~~~~~ \\
    18 & --- & 0.407~~~~~ \\
    21 & --- & 1.710~~~~~ \\
    24 & --- & 14.967~~~~~ \\
    \hline
\end{tabular}
\caption{Average time in milliseconds to compute one affine 
envelope\label{tbl:aff_results}}
\end{center}
\end{table}

\section{Conclusion}
\label{sec:conclusion}
Approximation and the generation of envelopes for Boolean formulas is used 
extensively in the querying of knowledge bases. Previous research has focused 
on the use of Horn approximations represented in conjunctive normal form (CNF).
In this paper, following the suggestion of Zanuttini, we instead focused on
the class of affine functions, using an approximation algorithm suggested by
Zanuttini~\cite{Zan:ECAI02}.  Our initial implementation using a naive
sets-of-models (as arrays of bitstrings) representation was disappointing, as
even for functions with very few models, the affine envelope often has very
many models (in fact, the affine envelope of very many functions is \true),
so storing sets of models as an array becomes prohibitive even
for functions over rather few variables.

ROBDDs have proved to be an appropriate representation for many applications
of Boolean functions.  Functions with very many models, as well as very few,
have compact ROBDD representations.  Thus we have developed a new affine
envelope algorithm using ROBDDs.  Our approach is based on the same principle
as Zanuttini's, but takes advantage of some useful characteristics of ROBDDs.
In particular, Propositions~\ref{prop:kill-var} and
\ref{prop:aff-exists-commute} allow us to
project away variables aggressively, often significantly reducing the sizes
of the representations being manipulated earlier than would happen otherwise.

Zanuttini~\cite{Zan:ECAI02} suggests an affine envelope algorithm using
modulo 2 congruence equations as output, and proves a polynomial complexity
bound.  However, we preferred to use ROBDDs.  As a functionally complete
representation for Boolean functions, ROBDDs allow the same representation
for input and output, keeping the algorithms simple.  
For example, the algorithm for evaluating whether one
ROBDD entails another is very straightforward, whereas evaluating whether a
set of congruence equations entails a Boolean function in some other
representation would be more complicated.  It also means that systems
which repeatedly construct an affine approximation, then manipulate it as a
general Boolean function, and then approximate this again, can operate
without having to repeatedly convert between different representations.
Importantly for our purposes, computing envelopes as ROBDDs permits us
to use the same representation for approximation to many different Boolean
classes.

Further research in this area includes implementing Zanuttini's suggested
modulo 2 congruence equations representation and comparing to our ROBDD
implementation.  This also includes evaluating the cost of determining
whether a set of congruence equations entail a given general Boolean
function.  We also will compare affine approximation to approximation to
other classes for information loss to evaluate whether affine functions
really are as suitable for knowledge-base approximations as Horn or other
functions.

\bibliographystyle{plain}
\bibliography{robdd}

\begin{thebibliography}{10}

\bibitem{bra-rud-dac-90}
K.~Brace, R.~Rudell, and R.~Bryant.
\newblock Efficient implementation of a {BDD} package.
\newblock In {\em Proc.\ Twenty-seventh ACM/IEEE Design Automation Conf.},
  pages 40--45, 1990.

\bibitem{bry-toc-86}
R.~Bryant.
\newblock Graph-based algorithms for {Boolean} function manipulation.
\newblock {\em IEEE Trans.\ Computers}, C--35(8):677--691, 1986.

\bibitem{Bryant:CS92}
R.~Bryant.
\newblock Symbolic {Boolean} manipulation with ordered binary-decision
  diagrams.
\newblock {\em ACM Computing Surveys}, 24(3):293--318, 1992.

\bibitem{Dec-Pea:AI92}
R.~Dechter and J.~Pearl.
\newblock Structure identification in relational data.
\newblock {\em Artificial Intelligence}, 58:237--270, 1992.

\bibitem{delVal:AI2005}
A.~{del Val}.
\newblock First order {LUB} approximations: Characterization and algorithms.
\newblock {\em Artificial Intelligence}, 162:7--48, 2005.

\bibitem{Hor-Iba:AI2003}
T.~Horiyama and T.~Ibaraki.
\newblock Ordered binary decision diagrams as knowledge-bases.
\newblock {\em Artificial Intelligence}, 136:189--213, 2002.

\bibitem{Post:func-compl41}
E.~L. Post.
\newblock {\em The Two-Valued Iterative Systems of Mathematical Logic}.
\newblock Princeton University Press, 1941.
\newblock Reprinted in M. Davis, Solvability, Provability, Definability: The
  Collected Works of {Emil L. Post}, pages 249--374, Birkha{\"u}ser, 1994.

\bibitem{Sch-Son:VMCAI06}
P.~Schachte and H.~S{\o}ndergaard.
\newblock Closure operators for {ROBDDs}.
\newblock In E.~A. Emerson and K.~Namjoshi, editors, {\em Proc.\ Seventh Int.\
  Conf.\ Verification, Model Checking and Abstract Interpretation}, volume 3855
  of {\em Lecture Notes in Computer Science}, pages 1--16. Springer, 2006.

\bibitem{Sch-Son:SARA07}
P.~Schachte and H.~S{\o}ndergaard.
\newblock Boolean approximation revisited.
\newblock In I.~Miguel and W.~Ruml, editors, {\em Abstraction, Reformulation
  and Approximation: Proc.\ {SARA} 2007}, volume 4612 of {\em Lecture Notes in
  Artificial Intelligence}, pages 329--343. Springer, 2007.

\bibitem{Schaefer1978:complexity}
T.~J. Schaefer.
\newblock The complexity of satisfiability problems.
\newblock In {\em Proc.\ Tenth Ann.\ ACM Symp.\ Theory of Computing}, pages
  216--226, 1978.

\bibitem{Sel-Kau:JACM96}
B.~Selman and H.~Kautz.
\newblock Knowledge compilation and theory approximation.
\newblock {\em Journal of the {ACM}}, 43(2):193--224, 1996.

\bibitem{Zan:ECAI02}
B.~Zanuttini.
\newblock Approximating propositional knowledge with affine formulas.
\newblock In {\em Proceedings of the Fifteenth European Conference on
  Artificial Intelligence (ECAI'02)}, pages 287--291. IOS Press, 2002.

\bibitem{Zan:SARA02}
B.~Zanuttini.
\newblock Approximation of relations by propositional formulas: Complexity and
  semantics.
\newblock In S.~Koenig and R.~Holte, editors, {\em Proceedings of SARA 2002},
  volume 2371 of {\em Lecture Notes in Artificial Intelligence}, pages
  242--255. Springer, 2002.

\end{thebibliography}

\end{document}